\documentclass[
notitlepage
,floatfix,
aps,
prx,
reprint,
superscriptaddress
]{revtex4-1}
\usepackage[utf8]{inputenc}

\usepackage[caption=false]{subfig}
\usepackage{graphicx} 
\usepackage{epstopdf}
\usepackage{array}
\usepackage{verbatim}
\usepackage{amsmath,amsfonts,amssymb,amscd}
\usepackage{amsthm}
\usepackage{tabularx}
\usepackage{stmaryrd}
\usepackage{enumerate}
\usepackage[ruled]{algorithm2e}
\usepackage{enumitem}
\usepackage{wasysym}
\usepackage{braket}
\usepackage{color}
\usepackage{natbib}
\usepackage{newtxtext,newtxmath}

\definecolor{darkblue}{rgb}{0,0,0.5}
\usepackage{hyperref}
\hypersetup{
    bookmarksnumbered=true, 
    unicode=false, 
    pdfstartview={FitH}, 
    pdftitle={}, 
    pdfauthor={}, 
    pdfsubject={}, 
    pdfcreator={}, 
    pdfproducer={}, 
    pdfkeywords={}, 
    pdfnewwindow=true, 
    colorlinks=true, 
    linkcolor=black, 
    citecolor=darkblue, 
    filecolor=blue, 
    urlcolor=darkblue 
}

\usepackage{refcount}

\theoremstyle{plain}
\newtheorem{thm}{Theorem}

\theoremstyle{definition}

\newcommand{\eq}[1]{(\hyperref[eq:#1]{\ref*{eq:#1}})}

\renewcommand{\sec}[1]{\hyperref[sec:#1]{Section~\ref*{sec:#1}}}
\newcommand{\thrm}[1]{\hyperref[thrm:#1]{Theorem~\ref*{thrm:#1}}}
\newcommand{\lemm}[1]{\hyperref[lemm:#1]{Lemma~\ref*{lemm:#1}}}
\newcommand{\prop}[1]{\hyperref[prop:#1]{Proposition~\ref*{prop:#1}}}
\newcommand{\corr}[1]{\hyperref[corr:#1]{Corollary~\ref*{corr:#1}}}
\newcommand{\fig}[1]{\hyperref[fig:#1]{~\ref*{fig:#1}}}
\newcommand{\deff}[1]{\hyperref[deff:#1]{~\ref*{deff:#1}}}

\newcommand{\ma}{\mathcal{A}}
\newcommand{\mE}{\mathcal{E}}

\newcommand{\mU}{\mathcal{U}}

\newcommand{\mbI}{\mathbb{I}}

\DeclareMathAlphabet{\matheu}{U}{eus}{m}{n}

\DeclareMathOperator{\Tr}{Tr}

\newcommand{\ketbra}[2]{|{#1}\rangle\!\langle{#2}|}

\newcommand{\ba}{\begin{eqnarray}}
\newcommand{\ea}{\end{eqnarray}}
\newcommand{\bann}{\begin{eqnarray*}}
\newcommand{\eann}{\end{eqnarray*}}
\newcommand{\dm}[1]{\ketbra{#1}{#1}}

\DeclareFontFamily{U}{mathc}{}
\DeclareFontShape{U}{mathc}{m}{it}%
{<->s*[1.03] mathc10}{}
\DeclareMathAlphabet{\mathscr}{U}{mathc}{m}{it}

\newcolumntype{L}[1]{>{\raggedright}p{#1}}
\newcolumntype{C}[1]{>{\centering}p{#1}}
\newcolumntype{R}[1]{>{\raggedleft}p{#1}}
\newcolumntype{D}{>{\centering\arraybackslash}X}

\usepackage{xr}

\newcommand{\vb}{\vspace*{\baselineskip}}

\begin{document}

\title{Skew informations from an operational view via resource theory of asymmetry}

\author{Ryuji Takagi}
\email{rtakagi@mit.edu}
\affiliation{Center for Theoretical Physics and Department of Physics, Massachusetts Institute of Technology, Cambridge, Massachusetts 02139, USA}
  
\begin{abstract}
The Wigner-Yanase skew information was proposed to quantify the information contained in quantum states with respect to a conserved additive quantity, and it was later extended to the Wigner-Yanase-Dyson skew informations. Recently, the Wigner-Yanase-Dyson skew informations have been  recognized as valid resource measures for the resource theory of asymmetry, and their properties have been investigated from a resource-theoretic perspective. The Wigner-Yanse-Dyson skew informations have been further generalized to a class called metric-adjusted skew informations, and this general family of skew informations have also been found to be valid asymmetry monotones.  Here, we analyze this general family of the skew informations from an operational point of view by utilizing the fact that they are valid asymmetry resource monotones. We show that such an approach allows for clear physical meanings as well as simple proofs of some of the basic properties of the skew informations. Notably, we constructively prove that any type of skew information cannot be superadditive, where the violation of the superadditivity had been only known for a specific class of skew informations with numerical counterexamples. We further show a weaker version of superadditivity relation applicable to the general class of the skew informations, which proves a conjecture made for the Wigner-Yanase skew information as a special case. We finally discuss an application of our results for a situation where quantum clocks are distributed to multiple parties. 

\end{abstract}


\maketitle
\section{Introduction} 
Quantifying the information contents is a central theme in information theory. 
In classical information theory, the Shannon entropy serves as such an information-theoretic quantity. 
In quantum information theory, the corresponding information measure is the von Neumann entropy, and it successfully reflects the total information that a state in a system possesses. 
However, quantifying the information contents becomes subtle when the system has a certain symmetry and possesses a conserved quantity because, in such cases, some observables can be measured more easily than others as observed by Wigner, Araki, and Yanase~\cite{wigner1952physik,Araki1960way,Yanase1961way}.
It is then natural to consider an information measure that takes into account the relation between a state and the conserved quantity.
Motivated by this observation, Wigner and Yanase proposed the information measure called the Wigner-Yanase skew information, which measures the information contents contained in a quantum state with respect to a conserved additive quantity~\cite{Wigner1963}. 
This information-theoretic quantity was later generalized by Dyson into one-parameter family called Wigner-Yanase-Dyson skew informations, and their properties were intensively investigated~\cite{Wigner1963,Lieb1973advances,Lieb1973prl,Wehrl1978entropy,MONDAL2016speed}.

The connection was observed between the skew informations and Riemannian metrics studied in information geometry~\cite{Gibilisco2003,Hansen2008}.
The family of quantum Fisher informations were identified as the metrics that reflect the natural monotonicity property under information processing, and one-to-one correspondence between quantum Fisher informations and operator monotone functions was established~\cite{Morozova1991,PETZ199681monotone}.
Based on the observation on the relation between the skew informations (information-theoretic measure) and the quantum Fisher informations (information metric), Hansen proposed a general family of skew informations parameterized by operator monotone functions originated from the information geometry, which is known as metric-adjusted skew informations~\cite{Hansen2008}.
In this paper, we simply say ``skew informations'' to refer to the whole family of metric-adjusted skew informations.
This characterization of the skew informations by information geometry finds a further connection to the measures of asymmetry in the context of resource theories.

Resource theories are formal frameworks dealing with quantification and manipulation of intrinsic physical quantities, called resources, associated with given physical settings.
The generality of the resource-theoretic framework enables us to extract common features shared by a large class of the theories~\cite{Horodecki2013,brandao_2015,Liu2017,Gour2017,regula_2018,Anshu2018catalytic,lami_2018,takagi2018operational,li2018quantifying,takagi2019general,liu2019one} and also provides specific formalism depending on the interested physical quantities such as entanglement~\cite{plenio2007introduction,HOrodecki_review2009}, coherence~\cite{aberg2006quantifying,Baumgratz2014,Streltsov2017}, asymmetry~\cite{Gour2008,Marvian2016}, quantum thermodynamics~\cite{Brandao2013,Brandao_secondlaws2015}, non-Markovianity~\cite{wakakuwa2017operational}, magic~\cite{Veitch2014,howard_2017}, and non-Gaussianity~\cite{Genoni2008,Takagi2018,albarelli2018resource}. 
Resource theories are especially powerful when one is interested in separating precious resources and free objects, as well as in assessing operational significance of the resources~\cite{Piani2009,bae_2016,Piani2015,Napoli2016,Piani2016,bae2018more,takagi2018operational,takagi2019general}.

In particular, resource theory of asymmetry accounts for the capability of breaking the relevant symmetry possessed by the system. The reference frame that can break the symmetry is treated as resource, and this setting turns out to be especially relevant to quantum metrology~\cite{Giovannetti2006}.
Considering the close connection between metrology and the quantum Fisher informations, and the connection mentioned above between skew informations and the information geometry, it is not surprising that there is also a connection between the skew informations and the resource theory of asymmetry. 
Indeed, it has been found that the Wigner-Yanase-Dyson skew informations~\cite{Marvian2012, marvian2014extending} as well as the whole family of skew informations~\cite{Zhang2017} serve as valid asymmetry quantifiers, which give another operational aspect to this information-theoretic quantities. 

Due to the generality of the skew informations as well as restrictions imposed on them, investigation of mathematical properties of the skew informations usually requires highly involved mathematical techniques~\cite{Lieb1973advances,Hansen2008} that are not physically very intuitive, and it is hoped that the operational view stemming from the resource theory would provide another route that gets around with these difficulties.
Indeed, such an approach has been developed in Ref.\,\cite{Marvian2012}, which, in particular, showed the selective monotonicity of the Wigner-Yanase-Dyson skew informations by seeing them as asymmetry monotones and also observed that extensive quantities can be freely amplified by a covariant operation, which gave an intuitive operational explanation of the violation of the uncertainty relation in terms of the Wigner-Yanase skew information proposed in Refs.\,\cite{Luo2003uncertainty,luo2004informational,luo2004skew}.
It has been also shown that any functions of the Noether's conserved quantities cannot be asymmetry monotones (see also \cite{marvian2014extending}). 

Here, we employ an operational approach to analyze properties of the general class of skew informations and see that such an operational point of view allows for richer physical intuitions and simpler proofs of them.
Notably, we constructively show that any skew information cannot be superadditive. 
The superadditivity of the Wigner-Yanase skew information was listed as a desired property for the skew information to be an information measure, and Wigner and Yanase themselves proved this property for pure bipartite states~\cite{Wigner1963}. 
It had been widely believed that it would hold in general until counterexamples were found~\cite{Hansen2007,Seiringer2007,Cai2008}.
Although it would be interesting to investigate what property of the state contributes to the violation of the superadditivity, the previously shown counterexamples are purely numerical examples obtained by exhaustive computational search or semianalytical forms which fail to provide much physical insights. 
In this work, instead of taking a counterexample-based argument, we utilize the fact that the skew informations are asymmetry monotones and find that the violation of the superadditivity is a natural consequence from the resource-theoretic point of view.
We also propose and prove a weaker version of superadditivity relation that holds for any skew information, which extends the results in Refs.\,\cite{Cai2008,Cai2010}. 
We employ an operational argument to show that our inequality is optimal, which proves the conjecture proposed in Ref.\,\cite{Cai2008} as a special case. Our results are then applied to a physical situation where quantum clocks are distributed to multiple parties. 

This paper is organized as follows. 
In Sec.\,\ref{sec:resource theory} and Sec.\,\ref{sec:skew info}, we briefly review the resource theory of asymmetry and the skew informations. In Sec.\,\ref{sec:properties}, we discuss some properties of the skew informations for which operational approach turns out to be helpful. In Sec.\,\ref{sec:superadditivity}, after reviewing a protocol used in the following discussion, we prove the violation of superadditivity of the skew informations and propose a weak superadditivity relation, followed by an application of these results to distributed quantum clocks.
We finally conclude the paper in Sec.\,\ref{sec:conclusions}. 

\section{Resource theory of asymmetry}
\label{sec:resource theory}
Main building blocks of resource theories include the sets of free states and free operations, which represent free objects that are considered to be provided at no cost. 
The resource theory of asymmetry with group $G$ corresponds to the setting where one has free access to quantum states that are invariant (symmetric) under group action whereas states that can break the group symmetry, asymmetric states, are considered precious, and thus resources.
Formally, the state $\rho$ is called a symmetric state if $U_g\rho U_g^{\dagger}=\rho,\ \forall g\in G$ where $U_g$ is a unitary representation of the group element $g$.  
A relevant set of free operations are covariant operations $\mE$ satisfying the covariance condition: $\mE\circ \mU_g^I = \mU_g^{O}\circ \mE\ \ \forall g\in G$ where $\mU_g^X(\cdot)$ refers to the application of unitary representation of $g$ on the system $X$, and $X=I,O$ correspond to the input system and output system of $\mE$. 
It can be easily seen that the covariant operations cannot create any asymmetry from symmetric states by noting that $\mU_g^O(\mE(\sigma))=\mE(\mU_g^I(\sigma))=\mE(\sigma)$ for any symmetric state $\sigma$, which is desired for the covariant operations to be free operations.  

The other important concepts in resource-theoretic frameworks are resource quantifiers. They are also called resource monotones because reasonable resource quantifiers $R$ must satisfy the monotonicity condition: $R(\rho)\geq R(\mE(\rho)),\ \forall \rho$ for any free operation $\mE$. 

Although the formalism of the theory of asymmetry encompasses a general choice of group, here we focus on the $U(1)$ group whose unitary representation is labeled by a real number $t$ as $U_X(t)=\exp(iH_Xt)$ where $H_X$ is an observable defined on system $X$.
This choice of group represents the type of coherence relevant to quantum metrology~\cite{Giovannetti2006}, quantum thermodynamics~\cite{Streltsov2017, Lostaglio2015, lostaglio2015description}, and quantum correlation~\cite{Girolami2013local,Ye2017oneway,Sun_2017from,Wang2019lower}.
As we shall see in the next section, a general family of skew informations, which originated from the information-theoretic motivation, have been shown to be valid asymmetry monotones in this case. 

\section{Skew informations}
\label{sec:skew info}
Suppose the system possesses an additive conserved quantity whose observable is denoted by $H$. 
To quantify the information contained by quantum states with respect to the conserved quantity, Wigner and Yanase proposed the Wigner-Yanase skew information~\cite{Wigner1963}: 
\ba
  I^{WY}(\rho,H) &=& -\frac{1}{2}\Tr([\sqrt{\rho},H]^2)\\
  &=& \Tr(\rho H^2) - \Tr(\sqrt{\rho}H\sqrt{\rho}H).
\ea

Later, the Wigner-Yanase skew information was generalized by Dyson, 
\ba
  I^{WYD}_{\alpha}(\rho,H) &=& -\frac{1}{2}\Tr([\rho^{\alpha},H][\rho^{1-\alpha},H])\\
  &=& \Tr(\rho H^2) - \Tr(\rho^{\alpha}H\rho^{1-\alpha}H).
\ea
with $0<\alpha<1$. Note that it reduces to the Wigner-Yanase skew information when one takes $\alpha=1/2$.
Remarkably, it was shown that the Wigner-Yanase-Dyson skew informations are valid asymmetry monotones~\cite{Marvian2012,marvian2014extending}.

The skew information is closely related to the information geometry, in which metrics represent `how close' the neighboring probability distributions or quantum states are in terms of their parameters~\cite{amari2007methods}.
For classical probability distribution, imposing the contractivity under information processing uniquely identifies the metric as the classical Fisher information~\cite{cencov2000statistical}.
In quantum theory, the contractivity does not single out the unique metric, but rather a family of metrics, the quantum Fisher informations, are specified~\cite{Morozova1991,PETZ199681monotone}. 
Consider the model where the quantum state is parameterized by a single real number $t$. 
Then, the quantum Fisher informations have the form
\ba
 J^f(\rho_t) = \Tr\left[\frac{\partial \rho_t}{\partial t}c_f(L_\rho,R_\rho) \left(\frac{\partial \rho_t}{\partial t}\right)\right]
 \label{eq:fisher_general}
\ea
where $c_f(x,y):=\left[yf(xy^{-1})\right]^{-1}$ is the Morozova-Chentsov function~\cite{Morozova1991}, and $f$ is a standard operator monotonic function satisfying 
\begin{enumerate}
\item 
$0\leq A \leq B$ implies $f(A) \leq f(B)$ for any Hermitian operators $A,B$.
\item
 $f(x) = xf(1/x)$
\item 
 $f(1)=1$,
\end{enumerate}
and $L_\rho,R_\rho$ are the superoperators multiplying $\rho$ from left and right: $L_{\rho} X = \rho X$ and $R_{\rho} X = X\rho$.
Based on the observation on the connection between Wigner-Yanase skew information and the quantum Fisher information with unitary model~\cite{Gibilisco2003}, a general family of skew informations called metric-adjusted skew informations were introduced~\cite{Hansen2008}:  
\ba
 I^f(\rho,H) = \frac{f(0)}{2}\Tr\left[ (i[H,\rho])\,c_f(L_\rho,R_\rho)\left(i[H,\rho]\right)\right]. 
 \label{eq:skew_general}
\ea
In this paper, we call this family of metric-adjusted skew informations simply skew informations. 
It can be explicitly calculated for the state $\rho=\sum_j \lambda_j \dm{j}$ as
\ba
 I^f(\rho,H)=\frac{f(0)}{2}\sum_{ij}\frac{(\lambda_i-\lambda_j)^2}{\lambda_jf(\lambda_i/\lambda_j)}|\bra{i}H\ket{j}|^2.
 \label{eq:skew_explicit}
\ea
Recently, it has been shown that the skew informations can be experimentally determined from linear-response theory~\cite{Shitara2016}.  
Importantly, the Wigner-Yanase-Dyson skew informations (and automatically also the Wigner-Yanase skew information) are special kinds of skew informations, which are reconstructed by taking
\ba
 c_f(x,y) = \frac{1}{\alpha(1-\alpha)}\frac{(x^\alpha-y^\alpha)(x^{1-\alpha}-y^{1-\alpha})}{(x-y)^2}
\ea
for $0<\alpha<1$.

An important property of the skew informations is that all the skew informations are valid asymmetry monotones~\cite{Zhang2017}, namely,
\ba
 I^f(\rho,H)\geq I^f(\mE(\rho),H)
\ea
for any covariant operation $\mE$.
It comes from the generic contractive property of the quantum Fisher metrics and the covariance of covariant operations with group transformations. 
Furthermore, the skew informations are additive for product states~\cite{Hansen2008}:
\ba
  I^f(\rho_1\otimes \rho_2,H_{12}) = I^f(\rho_1, H_1) + I^f(\rho_2, H_2)
 \ea
where $H_{12}=H_1\otimes \mbI + \mbI\otimes H_2$.

The additivity for product states is rather peculiar feature of skew informations among other asymmetry monotones; there are asymmetry monotones that are not additive for product states (e.g. relative entropy of asymmetry~\cite{Gour2009relative}).
 
\section{Properties of skew informations as asymmetry monotones}
\label{sec:properties}

As reviewed in the last section, the family of skew informations are defined in quite a general fashion through the operator monotone function $f$.
Its mathematical treatment could be cumbersome when the function has a complicated form or when one tries to keep its generality. 
Therefore, showing properties applicable to the general class of the skew informations could be highly non-trivial and mathematically involved while giving not much physical intuition. 

However, we have also seen that the skew informations serve as asymmetry monotones. 
This is an attractive property since it is not only valid for general skew informations but also gives a relevant physical meaning.  
Here, we show that such a resource-theoretic point of view may greatly simplify the analysis of some of the properties and help to give physical intuition associated with them. 

\subsection{Monotonicity under the partial trace}

The skew informations of the subsystem are not greater than the skew informations of the total system.
 \ba
 I^f(\rho_{12}, H_{12}) \geq I^f(\rho_1, H_1)
 \label{eq:ptrace}
 \ea
where $H_{12}=H_1\otimes \mbI + \mbI\otimes H_2$.
 As asymmetry measures, this relation entails a natural physical meaning; if one throws away a subsystem, the capability of breaking the symmetry must not increase. 
 Eq.\,\eqref{eq:ptrace} is concisely obtained by monotonicity of the skew informations under covariant operations.
 The key observation is that the partial trace is a covariant operation because 
 \ba
  &&\Tr_2\left[\exp(-iH_{12}t)\rho_{12}\exp(iH_{12}t)\right] \\
  \ \ &=&   e^{-iH_1t}\Tr_2\left[(\mbI\otimes e^{ -iH_2t})\rho_{12}(\mbI\otimes e^{ iH_2t})\right]e^{iH_1t}\\
  &=&\exp(-iH_1t)\Tr_2[\rho_{12}] \exp(iH_1t).
 \ea
Also, the following property holds as a special case of $H_2=0$:
\ba
 I^f(\rho_{12},H_1\otimes \mbI) \geq I^f(\rho_1, H_1).
\ea
Note that these properties have been discussed in the literature for only special cases such as Wigner-Yanase-Dyson skew information by explicitly calculating the quantities~\cite{Lieb1973advances,Li2011}.

\subsection{Convexity/selective monotonicity}
The results in this subsection were obtained in Ref.\,\cite{Zhang2017}, but we repeat them here because these are nice examples for which operational perspectives are helpful. 
The convexity of the skew informations has been shown in Ref.\,\cite{Hansen2008} where the L\"owner's theory of the operator monotone functions and analytic functions was actively used. 
The convexity can be more intuitively seen by staring from the monotonicity of the skew informations under covariant operations.
We combine the monotonicity under partial trace described above and the following relation 
\ba
  I^f\left(\sum_k p_k\rho_k\otimes \dm{k},H\otimes \mbI \right) = \sum_k p_k I^f(\rho_k, H),
  \label{eq:register}
 \ea
 where $\ket{k}$ is an orthonormal basis, which can be straightforwardly obtained using \eqref{eq:skew_explicit}.
 By using \eq{register} and \eq{ptrace}, we concisely reach the convexity property: 
\ba
 \sum_k p_k I^f(\rho_k, H) &=& I^f\left(\sum_k p_k\rho_k\otimes \dm{k},H\otimes \mbI\right) \\
 &\geq & I^f\left(\sum_k p_k\rho_k,H\right)
\ea
where we used the monotonicity under partial trace in the inequality. 
 Eq.\,\eqref{eq:register} can be also used to show the selective monotonicity. 
The selective monotonicity is the property of a resource measure that it does not increase on average. 
 Namely, for any operation $\mE=\sum_j\mE_j$ where $\mE_j$ is a covariant completely-positive trace non-increasing map for any $j$, the selective monotonicity states that  
\ba
 I^f(\rho,H)\geq \sum_k p_k I^f(\sigma_k, H)
\ea
where $p_k=\Tr(\mE_k(\rho))$ and $\sigma_k=\mE_k(\rho)/p_k$. 
The selective monotonicity is not regarded as a necessary property for resource quantifiers in general, but it is a reasonable feature shared by many important resource monotones~\cite{regula_2018}. 
This can be easily shown by considering another covariant operation $\tilde{\mE}(\cdot)=\sum_k\mE_k(\cdot)\otimes \dm{k}$: 
\ba
 I^f(\rho, H) &\geq& I^f(\tilde{\mE}(\rho), H\otimes \mbI)\\
 &=& I^f\left(\sum_k p_k\sigma_k\otimes \dm{k},H\otimes \mbI\right) \\
 &= & \sum_k p_k I^f\left(\sigma_k,H\right).
\ea
where we used the monotonicity under covariant operations in the inequality and \eqref{eq:register} in the second equality. 
 
\subsection{Decrease under measurements not disturbing the conserved quantity}
In Ref.~\cite{Luo2007decreases}, the dynamics of Wigner-Yanase skew information under the measurement that does not disturb the conserved quantity have been investigated. 
Specifically, they considered the measurement operation $M(\cdot)=\sum_j E_j\cdot E_j^{\dagger}$ whose measurement operators commute with the observable corresponding to the conserved quantity: $[E_j,H]=0,\ \forall j$. 
This implies $\Tr[M(\rho)H]=\Tr[\rho H]$, so the expectation value of the conserved quantity is not disturbed.  
For such a measurement, they asked whether the Wigner-Yanase skew information would decrease under deterministic measurement 
\ba
I^{WY}(M(\rho),H)\leq I^{WY}(\rho,H),
\label{eq:meas_deterministic}
\ea
and under selective measurement 
\ba
\sum_j p_jI^{WY}(\sigma_j,H)\leq I^{WY}(\rho,H)
\label{eq:meas_selective}
\ea
where $p_j =\Tr[E_j\rho E_j^{\dagger}]$ and $\sigma_j=E_j\rho E_j^{\dagger}/p_j$.
They proved that \eqref{eq:meas_deterministic} holds in general and proved \eqref{eq:meas_selective} holds for two-dimensional systems, while they left the higher dimensional cases as a conjecture. 

By our resource-theoretic approach, \eqref{eq:meas_deterministic} and \eqref{eq:meas_selective} are naturally shown for general dimensions and any skew information (not only for the Wigner-Yanase skew information). 
To see this, note that the condition $[E_j,H]=0,\ \forall j$ implies that $M$ is a covariant operation and $E_j\cdot E_j^{\dagger}$ are covariant completely-positive trace non-increasing maps. Then, \eqref{eq:meas_deterministic}, \eqref{eq:meas_selective}, and their generalizations with $I^f$ replacing $I^{WY}$ hold true as immediate consequences from the monotonicity and selective monotonicity of the skew informations as asymmetry monotones.

\subsection{Invariance under covariant unitaries }
Another immediate property of the skew informations as asymmetry monotones is the invariance under unitary that commutes with the observable, namely,
\ba
 I^f(\rho, H) = I^f(U\rho U^{\dagger}, H) 
 \label{eq:free_unitary}
\ea
for $U$ such that $[U,H]=0$. This commuting property means that $U$ is a free unitary in the resource theory of asymmetry. 

Eq.\,\eqref{eq:free_unitary} is due to a generic feature of asymmetry monotones; if \eqref{eq:free_unitary} did not hold, it would imply that only one-way transformation between $\rho$ and $U\rho U^{\dagger}$ would be possible under free operations. 
Since $U$ is a free unitary, clearly $\rho\rightarrow U\rho U^{\dagger}$ is possible.
However, since $[U,H]=0$ implies $[U^{\dagger},H]=0$, $U^{\dagger}$ is also free unitary, so $U\rho U^{\dagger}\rightarrow \rho$ is also possible by free unitary, which is a contradiction.

\section{Superadditivity}
\label{sec:superadditivity}
Superadditivity of the skew informations refers to the property that the skew informations for total states are never less than the sum of local skew informations. Formally, we say that the supearadditivity holds if for any $n\in \mathbb{N}$, 
\ba
   I^f(\rho_{1\dots n},H_{1\dots n}) \geq \sum_{k=1}^n I^f(\rho_k,H_k)
   \label{eq:superadditive multipartite}
\ea
holds for any $\rho_{1\dots n}$ and $H_{1\dots n}$ with $H_{1\dots n}=\sum_{k=1}^n H_k\otimes \mbI_{\bar{k}}$ where $\rho_k=\Tr_{\bar{k}}\rho_{1\dots n}$ is the reduced state on the $k$\,th subsystem, and $\mbI_{\bar{k}}$ is the identity operator acting on the subsystems other than $k$\,th subsystem.
An equivalent form of this is the superadditivity for general bipartite states
\ba
   I^f(\rho_{12},H_{12}) \geq I^f(\rho_1,H_1) + I^f(\rho_2,H_2)
   \label{eq:superadditive bipartite}
\ea
for any $\rho_{12}$ and $H_{12}$ with $H_{12}=H_1\otimes \mbI + \mbI\otimes H_2$.
The equivalence can be seen by observing that the former implies the latter by taking $n=2$, and the latter imples the former by
\bann
  I^f(\rho_{1\dots n},H_{1\dots n}) &\geq& I^f(\rho_{1},H_{1})+I^f(\rho_{2\dots n},H_{2\dots n})\\
  &\geq& I^f(\rho_{1},H_{1})+I^f(\rho_{2},H_{2})+I^f(\rho_{3\dots n},H_{3\dots n})\\
  &\dots&\\
  &\geq& \sum_{k=1}^n I^f(\rho_k,H_k)
\eann
where we sequentially used \eqref{eq:superadditive bipartite} to obtain \eqref{eq:superadditive multipartite}.
Note that it also means that the violation of \eqref{eq:superadditive multipartite} implies the violation of \eqref{eq:superadditive bipartite}.

Here, we employ an operational argument to show that for any choice of $f$, there exists $n$ for which \eqref{eq:superadditive multipartite} is violated for some $\rho_{1\dots n}$ and $H_{1\dots n}$.
It automatically leads to the violation of \eqref{eq:superadditive bipartite} for any skew information. 
The idea is that one can construct a covariant operation which creates larger sum of the local asymmetry than the global asymmetry.
For such a covariant operation, we consider the protocol proposed by \AA{}berg~\cite{Aberg2014}, and use it as a tool to prove the violation of the superadditivity~\footnote{For this purpose, one can also consider other covariant operations~\cite{Marvian2012,marvian2018nogo}}. 
Surprisingly, we also observe that the violation of the Wigner-Yanase skew information already occurs after two applications of \AA{}berg's protocol on the bipartite system.  
We finally prove a weaker version of superadditivity with a constant multiplied to one side of the inequality. We show that our inequality is optimal in terms of the multiplied constant where its optimality can be concisely shown by an operational argument.
Our result proves the conjecture proposed in Ref.\,\cite{Cai2008} as a special case.

\subsection{\AA{}berg's protocol}
We briefly review the protocol proposed by \AA{}berg, which may implement an asymmetric operation by applying a covariant operation over system and ancillary system with an asymmetric resource state~\cite{Aberg2014}. 
The resource state in the ancillary system works as a `catalyst' in the sense that the sequential use of the resource state does not decrease the accuracy of the implementation of the desired operation on the system. 

Let $S$ be a $d$-level system with the Hamiltonian $H_S=\sum_{j=0}^{d-1}j\ketbra{j}{j}$ and $A$ be an ancillary system with the Hamiltonian $H_A=\sum_{j\in \mathbb{Z}}j\ketbra{j}{j}$.
Suppose that we would like to implement some unitary $U$ on $S$ but only have access to global unitaries on $SA$ that commutes with the total Hamiltonian $H_t = H_S \otimes \mbI + \mbI \otimes H_A$.
Consider the following unitary $V(U)$ on $SA$,
\ba
 V(U) &=& \sum_{j,k=0}^{d-1} \ketbra{j}{j}U\ketbra{k}{k} \otimes \Delta^{k-j}
\label{eq:aberg_unitary}
\ea
where $\Delta = \sum_{j\in \mathbb{Z}}\ketbra{j+1}{j}$ is the operator that shifts the energy level in the ancillary system by one unit. 
It is easy to see that $V(U)$ commutes with the total Hamiltonian as required; $[V(U),H_t]=0$. 
Then, the quantum operation $\mE$ applied on $S$ is written by 
\ba
 \mE(\rho) &=& \Tr_A[V(U)\rho\otimes \sigma V(U)^{\dagger}]\\
 &=& \sum_{jklm} \Tr[\Delta^{-l+m+k-j}\sigma] U_{jk}\rho_{kl} (U^\dagger)_{lm} \ketbra{j}{m}.
\ea
where $O_{jk}\coloneqq\bra{j}O\ket{k}$ for an operator $O$.
Because of the commuting condition and that the partial trace is covariant, the quantum operation $\mE$ is also covariant.
Noting that the ideal implementation of $U$ gives $U\rho U^\dagger = \sum_{jklm} U_{jk}\rho_{kl} (U^\dagger)_{lm}\ketbra{j}{m}$, it can be seen that the accuracy of the implementation only depends on the values of $\Tr[\Delta^a\sigma]$ for $a=-2(d-1),\dots,2(d-1)$, and if $\Tr(\Delta^a \sigma)\sim 1$ for this range of $a$, $\mE(\rho) \sim U\rho U^{\dagger}$.
One good choice for $\sigma$ is $\sigma =\dm{\eta^M_l}$ where $\ket{\eta^M_l}=\frac{1}{\sqrt{M}}\sum_{i=0}^{M-1} \ket{i+l}$ for some $l$ and $M$. By taking $M$ large, $\mE(\rho)$ converges to $U\rho U^{\dagger}$.

What is remarkable about this protocol is that it is perfectly repeatable in the sense that one can reuse the resource state again and again without degrading the quality of the implementation of $U$. 
This can be seen by considering the reduced state on the ancillary system $A$ after one application of the protocol
\ba
 \sigma'=\Tr_S [V(U)\rho\otimes \sigma V(U)^{\dagger}].
\ea
It is straightforward to confirm that 
\ba
 \Tr[\sigma'\Delta^a] = \Tr[\sigma\Delta^a],\ \forall a.
 \label{eq:catalytic}
\ea
Since the effect of channel $\mE$ is only determined by $\Tr(\sigma \Delta^a)$, $\mE(\rho)$ remains the same although the reduced resource state gets changed. This perfect repeatability plays an important role in the following discussion.
Note also that one could introduce a lower bound for the spectrum of $H_A$ above to ensure the existence of the ground energy while keeping the perfect repeatability of the protocol by constantly supplying sufficient energy into the
ancillary system~\cite{Aberg2014}. Since such energy supply does not bring any asymmetry into the system, the following argument on the violation of superadditivity still holds.

\subsection{Violation of superadditivity}
Consider $N$ two-level systems $S_1,\dots,S_N$ and infinite-dimensional ancillary system $A$ with Hamiltonians $H_{S_i}=\dm{1}$ and $H_A=\sum_{j\in \mathbb{Z}} j\dm{j}$, and the initial state of the form
\begin{equation}
    \tilde{\rho}_0= \dm{0}^{\otimes N} \otimes \sigma_A
    \label{eq:aberg_initial}
\end{equation}
where $\sigma_A$ is an asymmetric resource state on $A$.
We choose this form for the sake of discussion, but the following argument can be easily extended to more general choice of systems and initial state.
Imagine that we sequentially apply the \AA{}berg's unitary \eqref{eq:aberg_unitary} on $S_m$ and $A$ at the $m$\,th step of the protocol. 
Let $V_{S_m A}$ be the unitary applied over $S_m$ and $A$, and define $\mathcal{V}_m(\cdot)=V_{S_m A}\cdot V_{S_m A}^{\dagger}$, and $\mathcal{V}_{1\dots k}=\mathcal{V}_k \circ \mathcal{V}_{k-1}\circ \cdots \circ \mathcal{V}_1$. 
Then, the reduced state on $S_1\dots S_N$ after $N$ applications of the unitary is 
\ba
 \tilde{\rho}_{f,S}= \Tr_A [\mathcal{V}_{1\dots N}(\tilde{\rho}_0)]=:\mathcal{A}(\tilde{\rho}_0)
 \label{eq:aberg_final}
\ea
where we defined $\mathcal{A}:= \Tr_A\circ \mathcal{V}_{1\dots N}$.
An important observation is that because of the commutation relation  $[V(U),H_t]=0$, $\mathcal{A}$ is a covariant operation, i.e. 
\ba
 \mathcal{A}(e^{-iH_t t}\rho e^{iH_t t}) = e^{-iH_S t} \mathcal{A}(\rho) e^{-iH_S t} 
\ea
where 
\ba
H_t = H_{S}\otimes \mbI + \mbI^{\otimes n}\otimes H_A,\ \ H_{S} = \sum_{j=1}^N H_{S_j}\otimes \mbI_{\bar{S}_j}. 
\ea

Since $\mathcal{A}$ is covariant, asymmetry of $\tilde{\rho}_{f,S}$ must be no larger than the asymmetry of $\tilde{\rho}_0$, and it holds true for any $N$ one chooses. 
Let us now take $I^f$ as an asymmetry monotone. Then, this observation leads to that $I^f(\tilde{\rho}_{f,S})\leq I^f(\tilde{\rho}_0)$. 
Furthermore, using the additivity of $I^f$ for product states and the fact that $\dm{0}$ is symmetric, we get $I^f(\tilde{\rho}_0)=I^f(\sigma_A)$.  
It implies that if we start with the resource state $\sigma_A$ with finite asymmetry (skew information), the asymmetry of the system is upper bounded by the initial asymmetry contained in $\sigma_A$.
However, if we take $U$ to be something that can create asymmetry, such as $U=\frac{1}{\sqrt{2}}\begin{pmatrix}
1 & 1 \\
1 & -1
\end{pmatrix}$, the reduced state of each subsystem gains non-zero asymmetry. Since the \AA{}berg's protocol is perfectly repeatable, one can take $N$ large enough that the sum of local asymmetry exceeds the initial asymmetry.   
This argument shows the following theorem on the violation of the superadditivity.
 \begin{thm}
  Any skew information $I^f$ cannot be superadditive.
  \label{thm:violation of superadditivity}
 \end{thm}
 
\begin{proof}
We explicitly construct the state violating the superadditivity by the above protocol. 
Consider the asymmetric resource state $\sigma_A$ with $I^f(\sigma_A)$ being finite.
Let $\tilde{\rho}_0$ be the initial state defined by \eqref{eq:aberg_initial} and $\tilde{\rho}_{f,S}$ be the final state obtained by \eqref{eq:aberg_final}. 
Then,
\ba
I^f(\sigma_A,H_A)=I^f(\tilde{\rho}_0, H_t)&\geq& I^f(\ma(\tilde{\rho}_0),H_{S})\\
&=&I^f(\tilde{\rho}_{f,S},H_{S})
\ea
where in the first equality, we used the additivity of the skew information for product states, and in the inequality, we used the monotonic property of the skew information as an asymmetry measure under covariant operations. 
Suppose, to the contrary, superadditivity holds. Then we must have
\ba
I^f(\tilde{\rho}_{f,S},H_{S})&\geq& \sum_{k=1}^N I^f(\tilde{\rho}_{f,S_k}, H_{S_k})\\
&=&NI^f(\tilde{\rho}_{f,S_1},H_{S_1})
\ea
where $\tilde{\rho}_{f,S_k}=\Tr_{\bar{S}_k}[\tilde{\rho}_{f,S}]$ is the reduced state on $S_k$.
The equality is because $\tilde{\rho}_{f,S_i}=\tilde{\rho}_{f,S_j}$ for all $i,j$ due to the perfect repeatability of the \AA{}berg's protocol.
However, since one can take $U$ such that $I^f(\tilde{\rho}_{f,S_1},H_S)>0$, $N$ can be taken large enough so that $NI^f(\tilde{\rho}_{f,S_1},H_S)>I^f(\sigma_A,H_A)$, which is a contraction.  
Hence, $I^f$ cannot be superadditive. 
\end{proof}

The above construction utilizes the fact that the \AA{}berg's protocol is perfectly repeatable, so we needed to consider multipartite setting where $N$ could be large. 
Interestingly, we find that application of the \AA{}berg's protocol on two qubits already shows the violation of the superadditivity for the Wigner-Yanase skew information. Here, we present a family of such bipartite states.

Let us set $U=\frac{1}{\sqrt{2}}\begin{pmatrix}
1 & 1 \\
1 & -1
\end{pmatrix}$ in \eq{aberg_unitary}  and   $\sigma_A=\ketbra{\eta_M^l}{\eta_M^l}$ where 
\begin{equation}
    \ket{\eta_M^l}=\frac{1}{\sqrt{M}}\sum_{i=0}^{M-1} \ket{i+l}.
\end{equation}
for some $M$ and $l$.
Roughly, $M$ specifies the amount of asymmetry in the ancillary system, and $l$ refers to the `position' of the coherence in the spectrum. 
When $N=2$, the total state after the protocol is
\begin{equation}
 \ket{\tilde{\psi}^{(2)}}=\frac{1}{2}(\ket{00}\ket{\eta_M^{l}}+\ket{01}\ket{\eta_M^{l-1}}+\ket{10}\ket{\eta_M^{l-1}}+\ket{11}\ket{\eta_M^{l-2}}).
\end{equation}

Taking the partial trace over $A$, 
\begin{eqnarray}
 \tilde{\rho}_{S_1S_2}&=&\Tr_A  \ketbra{\tilde{\psi}^{(2)}}{\tilde{\psi}^{(2)}} \nonumber \\
 &=& \frac{1}{4}
 \begin{pmatrix}
  1 && \alpha_1 && \alpha_1 && \alpha_2 \\
  \alpha_1 && 1 && 1 && \alpha_1 \\
  \alpha_1 && 1 && 1 && \alpha_1 \\
  \alpha_2 && \alpha_1 && \alpha_1 && 1
 \end{pmatrix}
\end{eqnarray}
 where $\alpha_i=\braket{\eta_M^{l-i}|\eta_M^{l}}=\frac{M-i}{M}$.
It gives 
\ba
 \tilde{\rho}_{S_1} = \tilde{\rho}_{S_2} = \frac{1}{2}\begin{pmatrix}
  1 & \frac{\alpha_1}{2} \\
  \frac{\alpha_1}{2} & 1
 \end{pmatrix}
\ea

Note that $\alpha_i$ only depends on $M$.
Fig.\,\ref{fig:superadditivity} shows the relation between $M$ and the degree of superadditivity $\Delta I^{WY} = I^{WY}(\tilde{\rho}_{12},H_{S_1S_2}) - \left( I^{WY}(\tilde{\rho}_{S_1},H_{S_1}) + I^{WY}(\tilde{\rho}_{S_2},H_{S_2}) \right)$. 
It can be seen that the superadditivity is violated for all $M\geq2$.

One can also see that the maximum violation occurs at $M = 4$, which gives $\alpha_1 = \frac{3}{4}$, $\alpha_2 = \frac{1}{2}$, and approaches 0 as $M$ increases. This asymptotic behavior is expected because at the limit of $M \to \infty$, the implementation of $U$ on each subsystem becomes perfect, and it brings $\tilde{\rho}_{S_1S_2} \to \ket{+}\ket{+}$. 
Since the skew information is additive for product state, $\Delta I^{WY} \to 0$ in this limit. 
It would be interesting to look into what contributes to the large violation; we leave the thorough analysis for future work.  

\begin{figure}[htbp]
 \includegraphics[scale = 0.4]{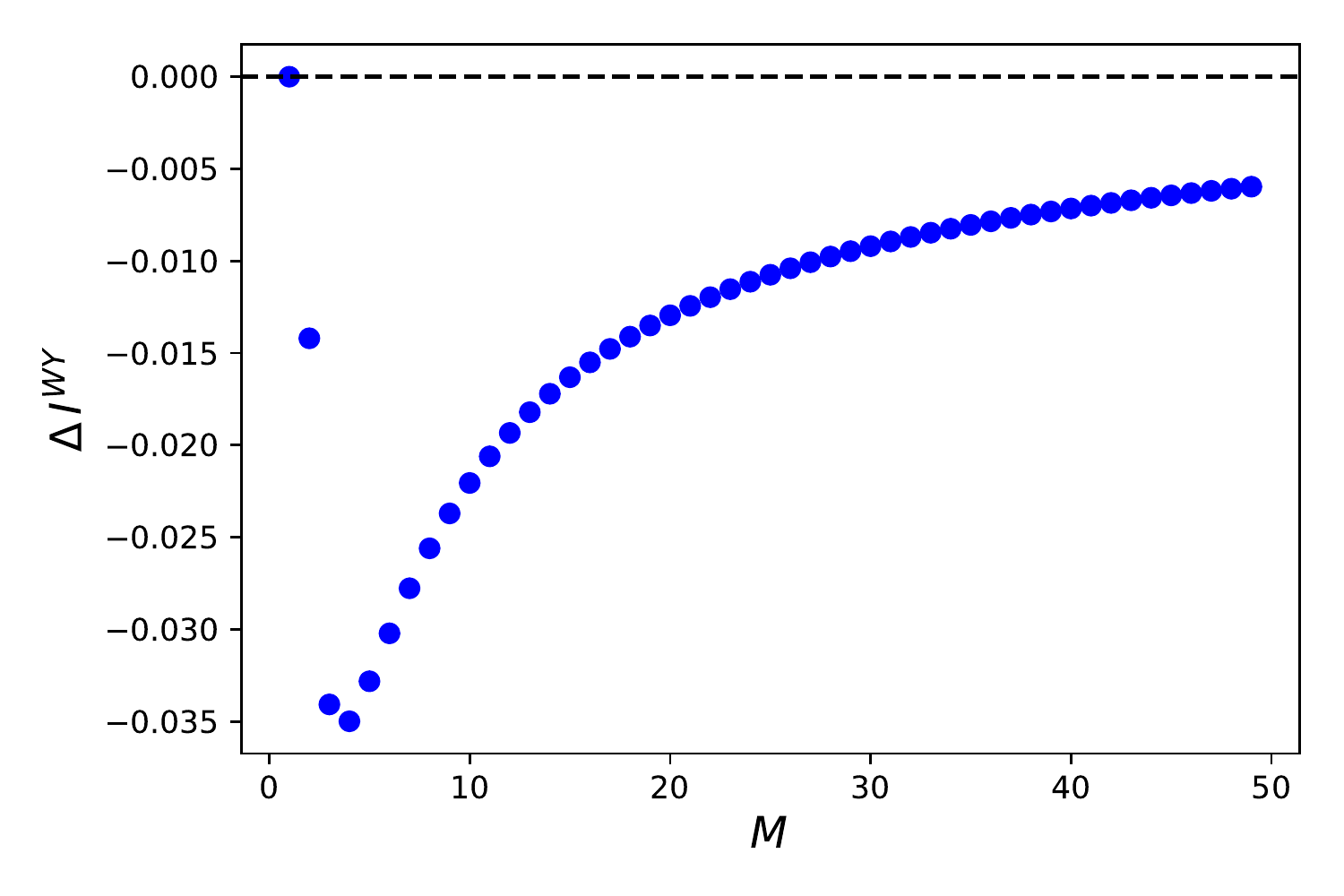}
 \caption{Relation between $M$ and $\Delta I^{WY} = I^{WY}(\tilde{\rho}_{S_1S_2},H_{S_1S_2}) - \left( I^{WY}(\tilde{\rho}_{S_1},H_{S_1}) + I^{WY}(\tilde{\rho}_{S_2},H_{S_2}) \right)$. Each data point corresponds to $\Delta I^{WY}$ for $M\in \mathbb{N},\ M\geq 1$.}
  \label{fig:superadditivity}
\end{figure}

\subsection{Weak superadditivity}
 Although the superadditivity does not hold in general, one can still ask whether some weaker version of superadditivity holds. 
 It has been shown that the following weak superadditivity for the Wigner-Yanase skew information holds~\cite{Cai2008}:
 \ba
 I^{WY}(\rho_{12},H_{12})\geq \frac{1}{2}(I^{WY}(\rho_1,H_1) + I^{WY}(\rho_2,H_2))
 \label{eq:weak_bi}
 \ea
 where $H_{12} = H_1 \otimes \mbI + \mbI \otimes H_2$.
They also conjectured that the factor $\frac{1}{2}$ is optimal in the sense that the largest constant $\beta$ such that  
 \ba
 I^{WY}(\rho_{12},H_{12})\geq \beta(I^{WY}(\rho_1,H_1) + I^{WY}(\rho_2,H_2))
 \ea
 holds for any $\rho_{12}$ and $H_{12}$ is $\frac{1}{2}$.
 
The following theorem further extends \eqref{eq:weak_bi} to general skew informations defined on any number of subsystems and identifies the optimal constant associated with it. 
The above conjecture is shown to be true as a special case of this result. 
We find that a similar resource-theoretic approach again provides a concise operational proof.

\begin{thm}
For any $\rho_{1\dots k}$ and $H_{1 2 \dots k} = \sum_j H_j\otimes \mbI_{\bar{j}}$, 
 \ba
 I^f(\rho_{1\dots k},H_{1 2 \dots k})\geq \frac{1}{k}\sum_{j=1}^{k} I^f(\rho_j,H_{j})
 \ea
 holds. Moreover, $\frac{1}{k}$ is the maximum constant for $\beta(k)$ such that
 \ba
 I^f(\rho_{1\dots k},H_{1 2 \dots k})\geq \beta(k)\sum_{j=1}^{k} I^f(\rho_j,H_{j})
 \ea
 holds for any $\rho_{1\dots k}$ and $H_{12\dots k}$.
 \label{thm:weak superadditivity}
\end{thm}
\begin{proof}
 Since partial trace is a covariant operation, monotonicity of the skew information under covariant operations implies 
 \ba
  \forall i,\ \ I^f(\rho_{1\dots k},H_{1 2 \dots k}) \geq I^f(\rho_i,H_i).
 \ea
By summing over $i$ and dividing both sides by $k$, we obtain
 \ba
 I^f(\rho_{1\dots k},H_{1 2 \dots k})\geq \frac{1}{k}\sum_{j=1}^{k} I^f(\rho_j,H_{j}).
 \label{eq:weak_multi}
 \ea
 We next show that the factor $1/k$ is optimal for any $k$.
 Suppose there exists $\beta'>\frac{1}{k}$ such that 
 \ba
 I^f(\rho_{1\dots k},H_{1 2 \dots k})\geq \beta'\sum_{j=1}^{k} I^f(\rho_j,H_{j})
 \label{eq:weak_super_fault}
 \ea
 holds for any $\rho_{1\dots k}$ and $H_{12\dots k}$ for some $k$.
 Let $N=k^n$ for some positive integer $n$, and $\rho_{1\dots N}$ and $H_{12\dots N}$ be state and Hamiltonian defined on the $N$ subsystems. We call these $N$ subsystems unit subsystems.
 Since \eqref{eq:weak_super_fault} should hold for any choice of state and Hamiltonian, we apply it to $k$ subsystems each of which consists of $k^{n-1}$ unit subsystems, which gives
 \ba
 I^f(\rho_{1\dots N},H_{1 2 \dots N})&\geq& \beta'\sum_{j=1}^{k} I^f(\rho_{1\dots k^{n-1}}^j,H_{1\dots k^{n-1}}^j)
 \ea
 where $\rho_{1\dots k^{n-1}}^j$ and $H_{1\dots k^{n-1}}^j$ denote the reduced state and Hamiltonian defined on the $j$\,th larger subsystem.
 We keep applying \eqref{eq:weak_super_fault} by dividing each subsystem into $k$ equal-sized subsystems until it reaches the unit subsystem. $n=\log_k N$ levels of division completes the task, and we end up with
 \ba
 I^f(\rho_{1\dots N},H_{1 2 \dots N})&\geq& (\beta')^{n}\sum_{j=1}^{N} I^f(\rho_j,H_{j})\\
 &=& \left(\frac{1}{N}\right)^{\log_k(1/\beta')}\sum_{j=1}^{N} I^f(\rho_j,H_{j}).
 \ea
 Note that $\log_k (1/\beta')<1$ due to the assumption that $\beta' > 1/k$. It implies that if the left hand side is finite value that does not depend on $N$, the sum of local asymmetry on the right hand side must grow sublinearly with $N$ because otherwise one could take sufficiently large $N$ that violates the inequality. 
 However, there exists a covariant operation that constructs a state with $\sum_{j=1}^{N} I^f(\rho_j,H_{j}) \propto N$ while $I^f(\rho_{S_1\dots N},H_{1 2 \dots N}) = \mathcal{O}(1)$ for any $N$. \AA{}berg's protocol is such an example since all the smallest subsystems have the identical marginal states because of the perfect repeatability while the final global asymmetry is finite because of the monotonicity of the skew informations and that the whole protocol is a covariant operation. 
 This is a contradiction, and hence we must have $\beta(k)\leq 1/k$. The equality is achieved due to \eqref{eq:weak_multi}.
\end{proof}

\subsection{Distributed quantum clocks}
Theorem \ref{thm:violation of superadditivity} and \ref{thm:weak superadditivity} provide an interesting implication for the situation where quantum clocks are distributed to multiple parties. 
Suppose that $k$ parties $A_1,A_2,\dots,A_k$ share some number of copies of state $\rho_{1\dots k}$ while each party only has access to their reduced state and does not know the description of the global state $\rho_{1\dots k}$. 
Assume also that they share a limited amount of entanglement among each other, which only enables them to send a limited number of qubits by quantum teleportation although they can freely make classical communication.
This is a situation relevant to the setups such as quantum network and distributed quantum computation~\cite{Duan2010qnetwork,Pirker2018qnetwork}.

From the perspective that the skew informations are asymmetry monotones, it is natural to see that they serve as resources for metrological tasks~\cite{Marvian2016}. In particular, when the conserved quantity is the Hamiltonian, the skew informations may be seen as relevant quantifiers for usefulness as quantum clocks~\cite{Janzing2003clock}.
Suppose that $A_1$ desires to possess a quantum clock with the amount of skew information $I^f_{\rm th}>I^f(\rho_1,H_1)$. In such a situation, $A_1$ could ask the other parties to send their states via quantum teleportation, but $A_1$ would like to make sure that it will be indeed possible to achieve the desired level of asymmetry by doing so because otherwise the precious entanglement will be wasted. 
To this end, suppose $A_1$ asks the other parties to measure the skew information of their own reduced state (by, for instance, a method provided by Ref.~\cite{Shitara2016}) and report it back by classical communication. $A_1$ then tries to infer the total skew information she would obtain $I^f(\rho_{1\dots k},H_{1,\dots,k})$ by reported values $I_f(\rho_j,H_j),\ j=2,\dots,k$.

Theorem \ref{thm:violation of superadditivity} warns $A_1$ not to make a naive decision in which she asks the other parties to send their states when $\sum_{j=1}^k I^f(\rho_j,H_j)\geq I^f_{\rm th}$ because it may be the case that $\rho_{1\dots k}$ significantly violates the superadditivity relation with $I^f(\rho_{1\dots k},H_{1,\dots,k}) < I^f_{\rm th}$. 
On the other hand, Theorem \ref{thm:weak superadditivity} ensures that if $A_1$ asks the other parties to send their states only when $\frac{1}{k}\sum_{j=1}^k I^f(\rho_j,H_j) \geq I^f_{\rm th}$, she will certainly obtain the enough amount of asymmetry. 
Moreover, it is the best possible she can do because $1/k$ is the maximum constant to ensure that $I^f(\rho_{1\dots k},H_{1,\dots,k}) \geq I^f_{\rm th}$ holds as shown in Theorem \ref{thm:weak superadditivity}.

\section{Conclusions}
\label{sec:conclusions}
We analyzed properties of the general family of skew informations from operational perspectives in the context of resource theory of asymmetry. 
We showed that such operational approach may give clearer physical meanings as well as simpler proofs of some of the properties of the skew informations.
We proved the violation of superadditivity for general family of skew informations by constructing a covariant operation that creates a state violating the superadditivity. 
Our proof has high contrast to the previous numerical-based approaches in that it can be applicable to general family of skew informations and also suggests a way of constructing the states showing the violation.
We observed that it is indeed a good ``violation producer'' by looking at the bipartite states produced by the protocol, which allowed us to provide a family of bipartite states violating the superadditivity of the Wigner-Yanase skew information. 
We also showed a weak superadditivity relation and proved the optimality of the inequality by an operational approach, which encompasses the previously postulated conjecture as a special case. 
We finally discussed an application of the violation of superadditivity and weak superadditivity relation proved in this work to a situation where quantum clocks are distributed to multiple parties, providing the optimal strategy for a single party to ensure that the enough amount of asymmetry will be obtained after costly quantum communications.  

It would be interesting to push this approach further and give clearer picture of what kind of property of states contributes to the violation of superadditivity since it would ultimately give insights into the genuine quantum nature of the quantum state in a system with conserved quantity.  
Our results indicate much potential of analyzing information-theoretic quantities from operational perspectives, and resource theories appear to be useful tools for that purpose. 
It would thus be intriguing to extend the analysis to a broader class of the quantities beyond the skew informations as well. 

\vb
--- {\it Note added.} 
Recently, we became aware of an independent related work by I. Marvian and R. W. Spekkens where they showed that no faithful asymmetry monotones can be subadditive or superadditive~\cite{marvian2018nogo}.

\begin{acknowledgements}
 The author thanks Iman Marvian and Kamil Korzekwa for fruitful discussions, Tomohiro Shitara for comments on the manuscript, and Seth Lloyd for making the publication of this work under open access possible. R.T. is supported by NSF, ARO, IARPA, and the Takenaka Scholarship Foundation. 
\end{acknowledgements}

\bibliographystyle{apsrmp4-2}
\bibliography{skew}

\end{document}